\documentclass[onecolumn,draftclsnofoot]{IEEEtran}
 \usepackage{cite}
\usepackage{amsmath,amssymb,amsfonts}
\usepackage{algorithmic}
\usepackage{graphicx}
\usepackage{subcaption}
\usepackage{textcomp}
\usepackage{xcolor}
\usepackage{mathtools} 
\usepackage{booktabs} 
\usepackage{tcolorbox}
\usepackage{tikz} 
\usepackage{color}
\usepackage{siunitx}
\usepackage{hyperref}
\usepackage{amsfonts}
\usepackage{bbold}
\usepackage{tikz-cd}
\usetikzlibrary{shapes,snakes,shapes.geometric}
\usetikzlibrary{arrows,positioning}
\usepackage[linesnumbered,algoruled,boxed,lined,ruled]{algorithm2e}
\usetikzlibrary{fit,calc}
\usepackage{enumitem}

\let\oldnl\nl
\newcommand{\nonl}{\renewcommand{\nl}{\let\nl\oldnl}}

\usepackage{amssymb,amsmath,amsthm,enumitem}
\newtheorem{problem}{Problem}

\newtheorem{remark}{Remark}
\newtheorem{theorem}{Theorem}
\newtheorem{definition}{Definition}

\newtheorem{lemma}{Lemma}

\usepackage{setspace}
 \newcommand{\citep}[1]{\cite{#1}}
 
\begin{document}
\title{Learning Circular Hidden Quantum Markov Models: A Tensor Network Approach}

\author{\IEEEauthorblockN{Mohammad Ali Javidian, \and Vaneet Aggarwal, \and  Zubin Jacob }\\
\IEEEauthorblockA{\textit{School of Electrical and Computer Engineering} \\
\textit{Purdue University}\\
\{mjavidia, vaneet, zjacob\}@purdue.edu}
}

\maketitle
\begin{abstract}
In this paper, we propose circular  Hidden Quantum Markov Models (c-HQMMs), which can be applied for modeling temporal data in quantum datasets (with classical datasets as a special case). We show that c-HQMMs are  equivalent to a constrained tensor network (more precisely, circular Local Purified State with positive-semidefinite decomposition) model. This equivalence enables us to provide an efficient learning model for c-HQMMs. The proposed learning approach is evaluated on six real datasets and demonstrates the advantage of c-HQMMs on multiple datasets as compared to  HQMMs, circular HMMs, and HMMs.
\end{abstract}

\section{Introduction}\label{sec:intro}
Hidden Markov Models (HMMs) are commonly used for modeling temporal data, usually, in cases where the underlying probability distribution is unknown, but certain output observations are known \citep{rabiner1986introduction,zucchini2009hidden}. Hidden Quantum Markov Models (HQMMs) \citep{Monras2010HiddenQM,clark2015hidden} can be thought of a reformulation of HMMs in the language of quantum systems (see section \ref{sec:related} for the formal definition). It has been shown that quantum
formalism allows for more efficient description of a given stochastic process  as compared to the classical case \citep{Monras2010HiddenQM,clark2015hidden,adhikary2021quantum}. In this paper, we propose circular HQMMs, and validate them to be more efficient than HQMMs.

We note that circular HMMs (c-HMMs) have been proposed to model HMMs, where the initial and terminal hidden states are connected through the state transition probability \citep{arica2000shape}.  c-HMMs have found application in speech recognition \citep{zheng1988text,shahin2006enhancing},  biology and meteorology \citep{holzmann2006hidden}, shape recognition \citep{arica2000shape,cai2007image}, biomedical engineering \citep{coast1990approach}, among others. Given the improved performance of c-HMMs as compared to HMMs, it remains open if such an extension can be done for HQMMs, which is the focus of this paper.

Even though multiple algorithms for learning HQMMs have been studied, direct learning of the model parameters is inefficient and ends up to poor local-optimal points \citep{adhikary2020expressiveness}. In order to deal with this challenge, a tensor network based approach is used to learn HQMMs \citep{adhikary2021quantum}, based on a result that HQMM is equivalent to uniform locally purified states (LPS) tensor network. The model in \citep{adhikary2021quantum} deals with infinite horizon HQMM, which involves uniform Kraus operators and thus uniform LPS. In this paper, we model a finite sequence of random variables, which allows us to have different Kraus operators at each time instant. We further extend the finite-horizon HQMMs to circular HQMMs (c-HQMMs).

In order to train the parameters of the c-HQMM, we show the equivalence of c-HQMM with a restricted class of tensor networks. In order to do that, we first define a class of tensor networks, called circular LPS (c-LPS). Then, we show that c-HQMM is equivalent to c-LPS where certain matrices formed from the decomposition are positive semi-definite (p.s.d.). Finally, we propose an algorithm to train c-LPS with p.s.d. restrictions, thus providing an efficient algorithm for learning c-HQMMs. 


The results in this paper show equivalence of finite-horizon HQMMs and c-HQMMs to the corresponding tensor networks. Further, we   show that c-HMM is equivalent to a class of tensor networks (circular Matrix Product State (c-MPS)) with non-negative real entries. This allows for an alternate approach of learning  c-HMMs, and may be of independent interest. 
The key contributions of this work are summarized as follows:

$\bullet$ We propose c-HQMM for modeling finite-horizon temporal data. 

$\bullet$ We show that c-HQMMs are  equivalent to c-LPS tensor networks with positive semi-definite matrix structure in the decomposition. Further, equivalent tensor structures for finite horizon HQMMs and c-HMMs are also provided. 

$\bullet$ Learning algorithm for c-HQMM is provided using the tensor network equivalence.

 In order to validate the proposed framework of c-HQMM and the proposed learning algorithm, we compare with standard HMMs (equivalent to MPS with non-negative real entries in decomposition), c-HMMs, and HQMMs. Evaluation on realistic datasets demonstrate the improved performance of c-HQMMs for modeling temporal data.

\section{Related Work and Background}\label{sec:related}
In this section, we briefly review the key related literature on hidden Markov models and tensor networks, with relevant definitions. 


{\em Hidden  Markov Models (HMMs)} \citep{rabiner1986introduction,zucchini2009hidden} are a class of probabilistic graphical models that have found greatest use in problems that enjoy an inherent temporality. These problems consist of a process that unfolds in time, i.e., we have states at time
$t$ that are influenced directly by a state at $t - 1$. 
HMMs have found application in such problems, for instance speech recognition \citep{juang1991hidden}, gesture recognition \citep{wilson1999parametric}, face recognition \citep{nefian1998hidden}, finance \citep{mamon2007hidden}, computational biology \citep{siepel2004combining,koski2001hidden,krogh1994hidden}, among others.  A finite-horizon hidden Markov model or HMM, as shown in Figure \eqref{fig:hmm}, consists of a discrete-time, discrete-state Markov chain, with hidden states $X_t\in\{1,\cdots,d\}, t\in\{1,\cdots,N\}$\footnote{Finite horizon implies that $N$ is finite, and thus the distributions can depend on the time-index. This paper focuses on finite horizon probability distributions. }, plus an observation model $p(o_t|x_t)$. The corresponding joint distribution has the form:
\begin{equation}
      p(X_{1:N},\mathcal{O}_{1:N})=p(X_{1:N})p(\mathcal{O}_{1:N}|X_{1:N})=\left[p(x_1)\Pi_{t=2}^{N}p(x_t|x_{t-1})\right]\left[\Pi_{t=1}^{N}p(o_t|x_t)\right]\label{eq:hmm} 
\end{equation}
\begin{figure*}[!ht]
\begin{subfigure}[t]{0.49\textwidth}
    \centering
    \includegraphics[width=\textwidth]{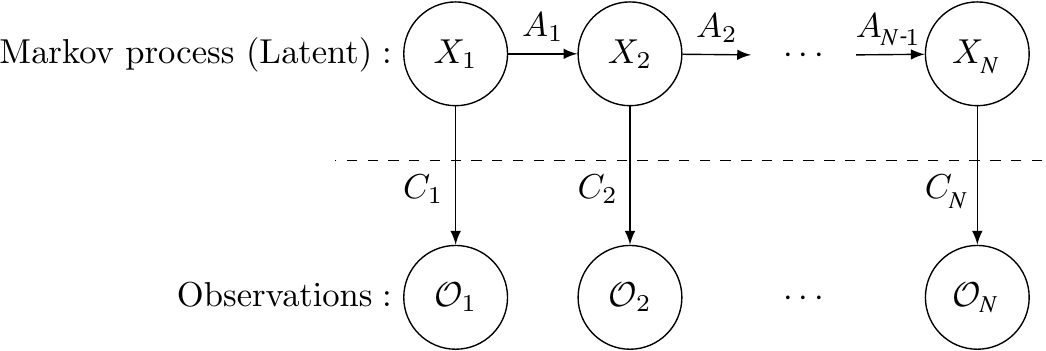}
    \caption{Representation of an HMM with observation variables $\mathcal{O}_i$, latent variables $X_i$, transition matrices $A_i$, and emission matrices $C_i$.}
    \label{fig:hmm}
\end{subfigure}
\hfill
\begin{subfigure}[t]{0.49\textwidth}
        \includegraphics[width=\textwidth]{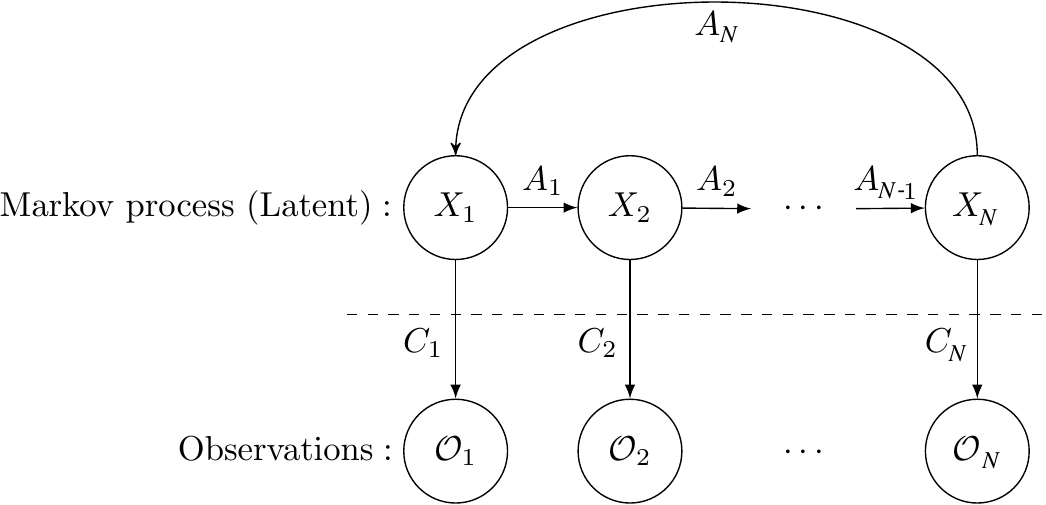}
    \caption{Representation of a c-HMM  with observation variables $\mathcal{O}_i$, latent variables $X_i$, transition matrices $A_i$, and emission matrices $C_i$.}
    \label{fig:chmm}
\end{subfigure}
\caption{Hidden Markov Model and Cyclic Hidden Markov Model.}
\end{figure*}

As shown in Figure \eqref{fig:hmm}, evolution of hidden states is governed by column-stochastic matrices $A_i$'s, called transition matrices, and emission matrices are column-stochastic matrices $C_i$'s that determine the observation probabilities. Despite the fact that HMMs are a powerful and versatile tool for statistical modeling of complex time-series data and stochastic dynamic systems, many real-world problems include circular data (e.g., measurements in the form of angles or other periodic values), such as biology, climatology, oceanography, geophysics, and astronomy. In these problems, the periodic nature of the boundary requires a Hidden Markov topology which is both \textit{temporal} (has a sequential order) and \textit{ergodic} to allow the revisits of a
state as the boundary returns to the starting point and
repeats itself \citep{arica2000shape}. c-HMMs were proposed to address these problems. Formally, a c-HMM is a  modification of  HMM model, where the initial and terminal hidden states are connected through the state transition probability $A_N$, as shown in Figure \eqref{fig:chmm}. The corresponding joint distribution has the form:
\begin{equation}
    p(X_{1:N},\mathcal{O}_{1:N})=p(X_{1:N})p(\mathcal{O}_{1:N}|X_{1:N})=\left[p(x_1|x_N)\Pi_{t=2}^{N}p(x_t|x_{t-1})\right]\left[\Pi_{t=1}^{N}p(o_t|x_t)\right]\label{eq:chmm}
\end{equation}

{\em Hidden Quantum Markov Model (HQMM)} was introduced in \citep{Monras2010HiddenQM} to model evolution from one quantum state to another, while generating classical output symbols. To produce an output symbol, a  measurement or Kraus operation \citep{kraus1983states} is performed on the internal state of the machine. To implement a Kraus operation, one can
use an auxiliary quantum system, called \textit{ancilla}. In every time step, the internal state of the HQMM interacts with its ancilla, which is then read out by a projective measurement. After every measurement, the ancilla is reset into its
initial state, while the internal state of the HQMM remains hidden \citep{clark2015hidden}. 
As in the classical case, an HQMM can be composed by the repeated
application of the quantum sum rule (plays the role of transition matrices in HMMs) and quantum Bayes
rule (plays the role of emission matrices in HMMs) \citep{adhikary2019learning} encoded using the sets of Kraus operators $\{K_{t,w}\}$ and $\{K_{t,x}\}$, respectively, for $t\in \{1, \cdots, N\}$: 
\begin{equation}
    \rho'_t=\sum_w K_{t,w}\rho_{t-1}K_{t,w}^\dagger \tag{quantum sum rule}
\end{equation}
\begin{equation}
    \rho_t=\frac{ K_{t,x}\rho'_{t}K_x^\dagger}{\textbf{\textrm{tr}}(\sum_x K_{t,x}\rho'_{t}K_{t,x}^\dagger)}\tag{quantum Bayes
rule}
\end{equation}
We can condense these two expressions into a single term for a given observation $x$ by setting $K_{t,x,w}=K_{t,x}K_{t,w}$, for $t=1,\cdots,N$:
\begin{equation}
    \rho_t=\frac{ \sum_w K_{t,x,w}\rho_{t-1}K_{t,x,w}^\dagger}{\textbf{\textrm{tr}}(\sum_w K_{t,x,w}\rho_{t-1}K_{t,x,w}^\dagger)}\tag{state update rule}
\end{equation}

We now formally define HQMMs using the Kraus operator-sum
representation (the definition is modified from \citep{adhikary2020expressiveness} to account for finite $N$).

\begin{definition}[HQMM]\label{def:hqmm}
An N-horizon d-dimensional Hidden Quantum Markov Model with a set of discrete
observations $\mathcal{O}$ is a tuple $(\mathbb{C}^{d\times d}, \{K_{i,x,w_x}\})$, where the initial state $\rho_0\in \mathbb{C}^{d\times d}$ and the Kraus operators $\{K_{i,x,w_x}\} \in \mathbb{C}^{d\times d}$, for all ${x\in \mathcal{O}},{i\in\{1,\cdots,N\}},w_x\in \mathbb{N}$, satisfy the following constraints:

$\bullet$ $\rho_0$ is a density matrix of arbitrary rank, and

$\bullet$ the full set of Kraus operators across all observables provide a quantum operation,\footnote{If $\sum_w\mathcal{K}_w^\dagger \mathcal{K}_w=I$, then $\mathcal{K}=\sum_w \mathcal{K}_w\rho\mathcal{K}_w^\dagger$ is called a quantum channel. However, if $\sum_w\mathcal{K}_w^\dagger \mathcal{K}_w<I$, then $\mathcal{K}$ is called a stochastic quantum operation.} i.e., $\sum_{x,w_x}K_{i,x,w_x}^\dagger K_{i,x,w_x}=I$ for all $i\in\{1,\cdots,N\}$. 

The joint probability of a given sequence is given by:
\begin{eqnarray}
{
    p(x_1,\cdots,x_N)=\vec{I}^T \Bigg(\sum_{w_{x_N}}K_{N,x_N,w_{x_N}}^\dagger\otimes K_{N,x_N,w_{x_N}}}\Bigg) \cdots\Bigg(\sum_{w_{x_1}}K_{1,x_1,w_{x_1}}^\dagger\otimes K_{1,x_1,w_{x_1}}\Bigg)\vec{\rho_0}\label{eq:hqmm}
\end{eqnarray}
where $\vec{I}^T,\vec{\rho_0}$ indicate vectorization
(column-first convention) of identity matrix and $\rho_0$, respectively. This model is illustrated in Figure \eqref{fig:hqmm}.
\end{definition}

This representation was used in \citep{srinivasan2018learning} to show that any $d$ dimensional HMM can be simulated as an equivalent $d^2$ dimensional HQMM \cite[Algorithm 1]{srinivasan2018learning}.


 HQMMs enable us to generate more complex
random output sequences than HMMs, even when using the same number of internal states \citep{clark2015hidden, srinivasan2018learning}. In other words, HQMMs are strictly more expressive than classical HMMs \citep{adhikary2021quantum}.

{\em Tensor Network}  is a set of tensors (high-dimensional arrays), where some or all of its indices are contracted according to some pattern \citep{oseledets2011tensor,orus2014practical}. They have been used in the study of many-body quantum systems \citep{orus2019tensor,montangero2018introduction}. Further, they  have been adopted for supervised learning in large-scale machine learning \citep{stoudenmire2016supervised,wang2017efficient,wang2018wide}. Some of the classes of tensor networks we use in this work include variants of Matrix Product States (MPSs) and Locally Purified States (LPSs).

One class of tensor networks is Matrix Product State (MPS), where an order-$N$ tensor $T_{d\times \cdots \times d}$, with rank $r$ has entry $(x_1,\cdots,x_N)$ ($x_i\in \{1, \cdots, d$) given as
\begin{equation}\label{eq:mps}
	T_{x_1,\cdots,x_N} = \sum_{\{\alpha_i\}_{i=0}^{N}=1}^r A_{0}^{\alpha_0}A_{1,x_1}^{\alpha_0,\alpha_1}A_{2,x_2}^{\alpha_1,\alpha_2}\cdots A_{N-1,x_{N-1}}^{\alpha_{N-2},\alpha_{N-1}}A_{N,x_{N}}^{\alpha_{N-1},\alpha_{N}}A_{N}^{\alpha_N}
\end{equation}
where $A_k, k\in\{0,N+1\}$, is a vector of dimension $r$, where element $(\alpha_k)$ is denoted as  $A_{k}^{\alpha_k}$. Further, $A_k, k\in\{1,\cdots,N\}$ is an order-3 tensor of dimension $d\times r\times r$, where element $(x,\alpha_L,\alpha_R)$ is denoted as  $A_{k,x}^{\alpha_{L},\alpha_R}$. 
\begin{figure*}[!ht]
    \centering
    \begin{subfigure}[t]{0.35\textwidth}
    \centering
        \includegraphics[width=\textwidth]{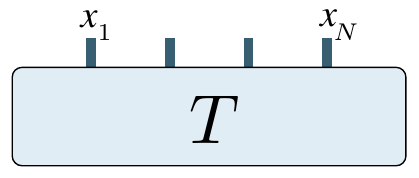}
        \caption{An order-$N$ Tensor $T$}
        \label{fig:ttrain}
    \end{subfigure}
    \hfill
    \begin{subfigure}[t]{0.35\textwidth}
    \centering
        \includegraphics[width=\textwidth]{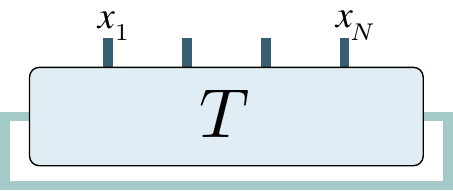}
        \caption{An order-$N$ circular Tensor $T$}
        \label{fig:tring}
    \end{subfigure}
    \hfill
    \begin{subfigure}[t]{0.49\textwidth}
    \centering
        \includegraphics[width=\textwidth]{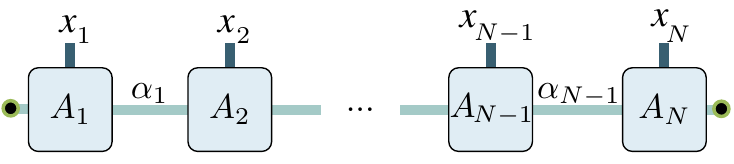}
        \caption{An order-$N$ MPS $T$}
        \label{fig:mps}
    \end{subfigure}
    \hfill
    \begin{subfigure}[t]{0.49\textwidth}
    \centering
        \includegraphics[width=\textwidth]{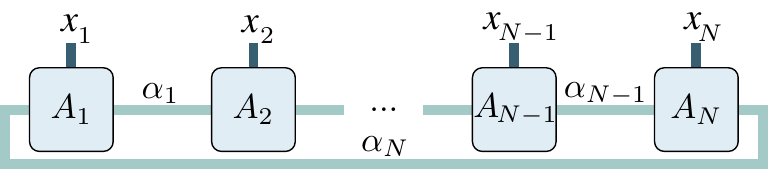}
        \caption{An order-$N$ c-MPS $T$.}
        \label{fig:cmps}
    \end{subfigure}
    
    \begin{subfigure}[t]{0.49\textwidth}
    \centering
        \includegraphics[width=\textwidth]{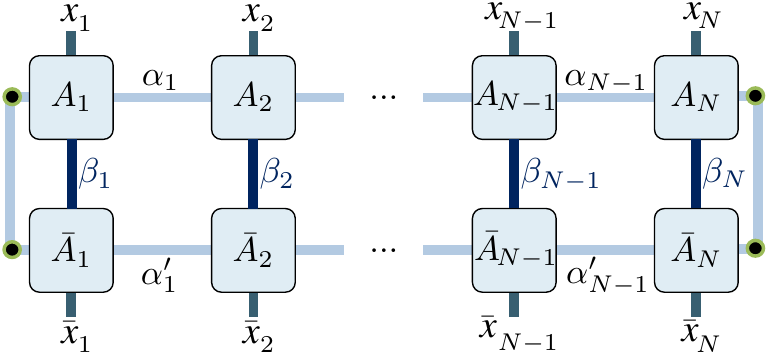}
        \caption{An order-$N$ LPS $T$}
        \label{fig:lps}
    \end{subfigure}
    \hfill
    \begin{subfigure}[t]{0.49\textwidth}
    \centering
        \includegraphics[width=\textwidth]{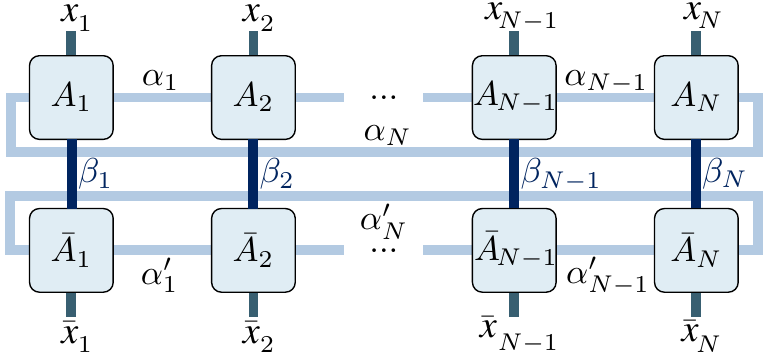}
        \caption{An order-$N$ c-LPS $T$}
        \label{fig:clps}
    \end{subfigure}
        \vspace{-.1in}
    \caption{Tensor diagrams corresponding to different tensor networks. Black end dots indicate boundary vectors.}\label{TNlabelfig}
    \vspace{-.1in}
\end{figure*}

Another class of tensor network that is studied in this paper is the Locally Purified State (LPS). An order-$N$ MPS with $d$-dimensional indices, admits an LPS representation of puri-rank $r$ and purification dimension $\mu$ when the entries of $T$ can be written as:
\begin{equation}\label{eq:lps}
	\begin{split}
		T_{x_1,\cdots,x_N} &= \sum_{\{\alpha_i,\alpha'_i\}_{i=0}^{N}=1}^r\sum_{\{\beta_i\}_{i=1}^N=1}^{\mu}A_0^{\alpha_0,\alpha'_0}A_{1,x_1}^{\beta_1,\alpha_0,\alpha_1}\overline{A_{1,x_1}^{\beta_1,\alpha'_0,\alpha'_1}}A_{2,x_2}^{\beta_2,\alpha_1,\alpha_2}\overline{A_{2,x_2}^{\beta_2,\alpha'_1,\alpha'_2}}\\ &\cdots A_{N-1,x_{N-1}}^{\beta_{N-1},\alpha_{N-2},\alpha_{N-1}}\overline{A_{N-1,x_{N-1}}^{\beta_{N-1},\alpha_{N-2},\alpha_{N-1}}} A_{N,x_{N}}^{\beta_N,\alpha_{N-1},\alpha_N}\overline{A_{N,x_{N}}^{\beta_N,\alpha'_{N-1},\alpha'_N}}A_{N+1}^{\alpha_N,\alpha'_N}
	\end{split}
\end{equation}
where $A_k, k\in\{0,N+1\}$, is an $r\times r$ matrix, where the element $(\alpha_k,\alpha'_k)$ is denoted as $A_{k}^{\alpha_k,\alpha'_k}$. Further, $A_k, k\in\{1,\cdots,N\}$, is an order-4 tensor of dimension $d\times \mu\times r\times r$, where the element $(x, \beta,\alpha_L,\alpha_R)$ is denoted as $A_{k,x}^{\beta,\alpha_L,\alpha_R}$, and elements belong to $\mathbb{R}$ or $\mathbb{C}$, as defined based on the context.

The tensor networks can be represented using tensor diagrams,  where tensors are represented by
boxes, and indices in the tensors are represented by lines emerging from the boxes. The lines connecting tensors between each other correspond to contracted indices, whereas lines that do not go from one tensor to another correspond to open indices \cite{orus2014practical}. The tensor diagrams corresponding to tensor networks MPS and LPS can be seen in Figure \eqref{fig:mps} and \eqref{fig:lps}, respectively.

\paragraph{Relation between HMMs and Tensor Networks: } As shown recently in \citep{Glasser2019,adhikary2021quantum}, tensor networks have direct correspondence with HMMs. In particular,  non-negative matrix product states (MPS) are HMMs \citep{Glasser2019}, and uniform locally purified states are QHMMs \citep{adhikary2021quantum}. Note that equivalence assumes that the tensor networks are normalized as the probabilities, while we will not explicitly normalize the tensor networks in the proofs, while will be accounted in the learning. 

\paragraph{Learning of HQMM: } Two state-of-the-art algorithms for learning HQMMs were proposed in \citep{srinivasan2018learning,adhikary2020expressiveness}. Both algorithms use an iterative maximum-likelihood algorithm to learn Kraus operators to model sequential data using an HQMM. The proposed algorithm in \citep{srinivasan2018learning} is slow and there is no theoretical gaurantee that the algorithm steps towards
the optimum at every iteration \citep{adhikary2020expressiveness}. The proposed algorithm in \citep{adhikary2020expressiveness}, however, uses a gradient-based algorithm. Although, the proposed algorithm in \citep{adhikary2020expressiveness} is able to learn an HQMM that outperforms the corresponding HMM, this comes
at the cost of a rapid scaling in the number of parameters. In order to deal with this issue, equivalence between HQMMs and Tensor Networks have been considered to achieve efficient learning \citep{Glasser2019,adhikary2021quantum}. 


\section{Proposed \texorpdfstring{\lowercase{c-}}HHQMM}

In this section, we propose Circular HQMM (c-HQMM) for modeling temporal data.

\begin{definition}[c-HQMM]\label{def:chqmm}
An $N$-horizon $d$-dimensional circular Hidden Quantum Markov Model (c-HQMM) with a set of discrete
observations $\mathcal{O}$ is a tuple $(\mathbb{C}^{d\times d}, \{K_{i,x,w_x}\},\textbf{\textrm{tr}}(\cdot))$, where the Kraus operators are given as $\{K_{i,x,w_x}\} \in \mathbb{C}^{d\times d}$, for all ${x\in \mathcal{O}},{i\in\{1,\cdots,N\}},w_x\in \mathbb{N}$. The full set of Kraus operators across all observables provide a quantum operation, i.e., $\sum_{i,x,w_x}K_{i,x,w_x}^\dagger K_{x,w_x}=I,$ for all $i\in\{1,\cdots,N\}$. 
The joint probability of a given sequence is given by:
\begin{equation}\label{eq:chqmm}
     p(x_1,\cdots,x_N)=\textbf{\textrm{tr}} \Bigg(\Bigg(\sum_{w_{x_N}}\Bar{K}_{N,x_N,w_{x_N}}\otimes K_{N,x_N,w_{x_N}}\Bigg)  \cdots\Bigg(\sum_{w_{x_1}}\Bar{K}_{1,x_1,w_{x_1}}\otimes K_{1,x_1,w_{x_1}}\Bigg)\Bigg)
\end{equation}
where $\textbf{\textrm{tr}}(\cdot)$ indicates the trace of the resulting matrix. This model is illustrated in Figure \eqref{fig:chqmm}.
\end{definition}

We note that this representation along with the algorithm proposed in  \cite[Algorithm 1]{srinivasan2018learning} can be used to show that any $d$ dimensional circular HMM can be simulated as an equivalent $d^2$ dimensional c-HQMM. 

\begin{figure*}[!ht]
\begin{subfigure}[t]{0.49\textwidth}
    \centering
    \includegraphics[width=\textwidth]{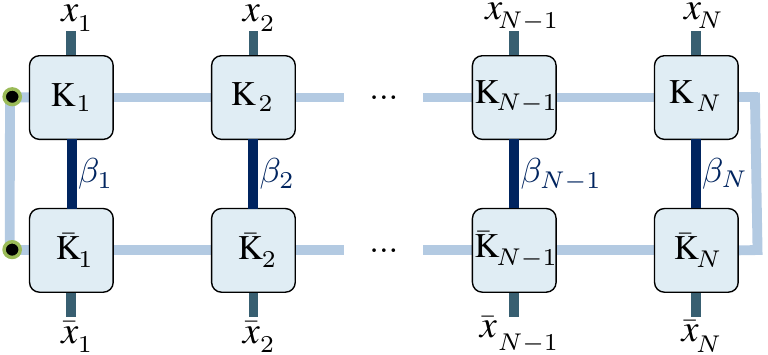}
    \caption{A Hidden Quantum Markov Model. The rightmost connecting line at the boundary represents the application of the identity. Black end dots indicate boundary vectors.}
    \label{fig:hqmm}
\end{subfigure}
\hfill
\begin{subfigure}[t]{0.49\textwidth}
        \includegraphics[width=\textwidth]{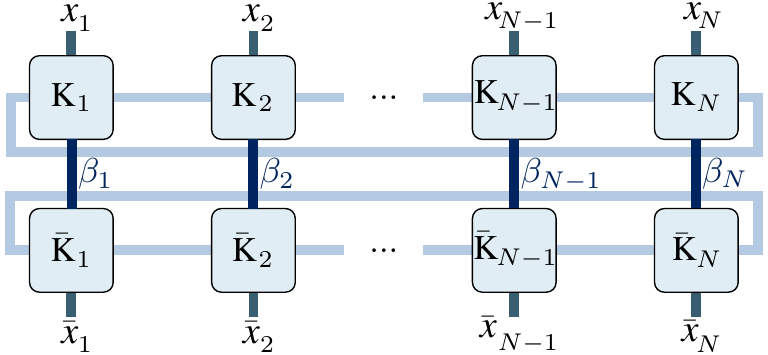}
    \caption{A Circular Hidden Quantum Markov Model.}
    \label{fig:chqmm}
\end{subfigure}
\caption{Hidden Quantum Markov Model and circular Hidden Quantum Markov Model: with observation variables ${x}_i$, where $K_i$ are $K_{i,x_i,w_{x_i}}$ and $\beta_i=|w_{x_i}|$ is determined by the Kraus-rank.}
\end{figure*}

\section{Proposed \texorpdfstring{\lowercase{c-}}-LPS Model}\label{sec:clps}

In this section, we propose circular LPS (c-LPS) which is an extension of LPS. For this purpose, we first briefly review circular MPS (c-MPS) model. Circular MPS (c-MPS) is an extension of MPS, where an order-$N$ c-MPS  $T$, with  $d$-dimensional indices and  rank $r$ has the entries given as:
\begin{equation}\label{eq:c-MPS}
    T_{x_1,\cdots,x_N} =\sum_{\{\alpha_i\}_{i=1}^N=1}^r A_{1,x_1}^{\alpha_{N},\alpha_1}A_{2,x_2}^{\alpha_1,\alpha_2}\cdots A_{N,x_{N}}^{\alpha_{N-1},\alpha_{N}}
\end{equation}
where $A_k, k\in \{1,\cdots,N\}$, is an order-3 tensors of dimension $d\times r\times r$, as shown in Figure \eqref{fig:cmps}, where element $(x,\alpha_L,\alpha_R)$ is denoted as  $A_{k,x}^{\alpha_{L},\alpha_R}$.
c-MPS are studied in the literature as tensor rings. Tensor rings \citep{zhao2016tensor,mickelin2020algorithms} have found application in compression of
convolutional neural networks \citep{wang2018wide}, image and video compression \citep{zhao2019learning}, data completion \citep{wang2017efficient}, among others.

Now, we introduce circular LPS ( c-LPS) as a tensor ring extension of LPS, where an order-$N$  with $d$-dimensional indices,  puri-rank $r$, and purification dimension $\mu$ has  entries given as:
\begin{equation}\label{eq:clps}
		T_{x_1,\cdots,x_N} =\sum_{\{\alpha_i,\alpha'_i\}_{i=1}^N=1}^r\sum_{\{\beta_i\}_{i=1}^N=1}^{\mu} A_{1,x_1}^{\beta_1,\alpha_{N},\alpha_1}\overline{A_{1,x_1}^{\beta_1,\alpha'_{N},\alpha'_1}}\cdots A_{N,x_{N}}^{\beta_N,\alpha_{N-1},\alpha_{N}}\overline{A_{N,x_{N}}^{\beta_N,\alpha'_{N-1},\alpha'_{N}}}
\end{equation}
where $A_k, k\in\{1,\cdots,N\}$, is an order-4 tensor of dimension $d\times \mu\times r\times r$, as shown in Figure \eqref{fig:clps}, where the element $(x, \beta,\alpha_L,\alpha_R)$ is denoted as $A_{k,x}^{\beta,\alpha_L,\alpha_R}$.\footnote{Since Born machines (BMs) are LPS of
purification dimension $\mu=1$, we can similarly define circular BMs.} 

\section{\texorpdfstring{\lowercase{c-}}-HQMM are \texorpdfstring{\lowercase{c-}}-LPS with Positive Semi-Definite Matrix Structure}\label{sec:relationship}

We first note that non-negative MPS (denoting by MPS$_{\mathbb{R}_{\ge 0}}$) are HMM \citep{Glasser2019}. In other words, any HMM can be mapped to an MPS with non-negative elements, and any MPS$_{\mathbb{R}_{\ge 0}}$ can be mapped to a HMM. Similarly, local quantum circuits with ancillas are locally purified states \citep{Glasser2019}. The authors of \cite{adhikary2021quantum} recently considered an infinite time model of HQMM, where  Kraus operators do not depend on time, and showed the equivalence of these HQMMs with the uniform LPS with a positive definite matrix structure. However, our work considers a non-uniform finite-time structure by having Kraus operators depend on time. We note that the equivalent tensor structure corresponding to c-HMM and c-HQMM are open, which is studied in this section. 

The next result describes the relation between c-HQMM and c-LPS:
\begin{theorem}\label{thm:chmmclps}
c-HQMM model is equivalent to a c-LPS structure where the decomposition entries $A_{i,x}^{b,a_1,a_2}$ are complex, and the $r\times r$ matrices formed by  $A_{i,x}^{b,\cdot,\cdot}$ for all $i,x,b$ are positive semi-definite (p.s.d.). 
\end{theorem}

We need the following lemmas to prove the theorem.  The first lemma shows that both c-HQMM and c-LPS models have the same form of operators and a c-LPS operator can be mapped to a c-HQMM operator under certain conditions.

\begin{lemma}\label{lem:tensor2kraus}
Consider a c-LPS of order-$N$  with $d$-dimensional indices,  puri-rank $r$, and purification dimension $\mu$ as defined in section 4, where the decomposition entries $A_{i,x}^{b,a_1,a_2}$ are complex, and the $r\times r$ matrices formed by  $A_{i,x}^{b,\cdot,\cdot}$ for all $i,x\in X_i,b\in\beta_i, i\in\{1,\cdots,N\}$ are positive semi-definite (p.s.d.). So, c-LPS has operators of the form $\tau_{x}=\sum_{\beta_i=1}^\mu\Bar{B}_{x,\beta_i}\otimes B_{x,\beta_i}$, where $\{B_{x,\beta_i}\}_{x\in X_i}\in \mathbb{C}^{r\times r}$. 
\end{lemma}

\begin{proof}
The given c-LPS, as defined in section 4, has entries of the form:
\begin{equation*}
	\begin{split}
		T_{x_1,\cdots,x_N} &=\sum_{\{\alpha_i,\alpha'_i\}_{i=1}^N=1}^r\sum_{\{\beta_i\}_{i=1}^N=1}^{\mu} A_{1,x_1}^{\beta_1,\alpha_{N},\alpha_1}\overline{A_{1,x_1}^{\beta_1,\alpha'_{N},\alpha'_1}}\cdots A_{N,x_{N}}^{\beta_N,\alpha_{N-1},\alpha_{N}}\overline{A_{N,x_{N}}^{\beta_N,\alpha'_{N-1},\alpha'_{N}}}
	\end{split}
\end{equation*}
For each $x_i\in\{1,\cdots,N\}$, the index contraction of $A_i$ and $\Bar{A}_i$ over $\beta_i\in\{1,\cdots,N\}$ results in an order-4 tensor of dimension $r\times r\times r\times r$. This 4-dimension array can be reshaped as an $r^2\times r^2$ matrix and rewritten as $\tau_{x_i}=\sum_{\beta_i=1}^\mu\Bar{B}_{x_i,\beta_i}\otimes B_{x_i,\beta_i}$, where $\{B_{x_i,\beta_i}\}_{x_i=1}^d\in \mathbb{C}^{r\times r}$. Because for the fixed values $x_i\in X_i$ and $\beta_i$, tensors $A_i, \Bar{A_i}$ can be reshaped and rewritten as $r\times r$ matrices $B_{x_i,\beta_i}, \Bar{B}_{x_i,\beta_i}$, respectively.
So, c-LPS has operators of the form defined as follows:

\centering\begin{tcolorbox}[hbox]
$
\begin {aligned}
\raisebox{-17mm}{\includegraphics{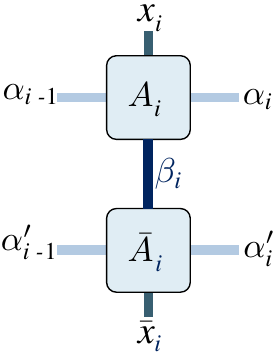}}
\qquad = \qquad
\sum_{\beta_i=1}^\mu\Bar{B}_{x_i,\beta_i}\otimes B_{x_i,\beta_i}
\end {aligned}
$
\end{tcolorbox}
\end{proof}

The following lemma shows that a c-HQMM operator can be mapped to a c-LPS operator.
\begin{lemma}\label{lem:Kraus2tensor}
Kraus operators of a c-HQMM model of rank $\beta$ can be mapped to operators of a c-LPS model of the same puri-rank. 
\end{lemma}

\begin{proof}
As shown in the proof of Lemma \ref{lem:tensor2kraus}, we know that both models have operators of the same
form.
Now, we map $\sum_{x_i,w_{x_i}}\Bar{K}_{i,x_i,w_{x_i}}\otimes K_{i,x_i,w_{x_i}}, i\in\{1,.\cdots,N\}$ to tensors $A_i, \Bar{A_i}, i\in\{1,.\cdots,N\}$ in an c-LPS. For this purpose, we set $\beta_i=|w_{x_i}|$, i.e., the Kraus-rank in c-HQMM plays the role of purification dimension in c-LPS. For fixed values $x_i$ and $w_{x_i}$, $K_{i,x_i,w_{x_i}}$ and $\Bar{K}_{i,x_i,w_{x_i}}$ are $d\times d$ matrices that play the role of tensors $A_i$ and $\Bar{A}_i$ in the c-LPS, respectively.
\end{proof}

The following lemma shows that the c-LPS transfer operator can be rescaled and transformed into a trace-preserving map.
\begin{lemma}\label{lem:tracepreserving}
The c-LPS transfer operator, i.e., $\tau=\sum_{x=1}^d\tau_{x}$ can be rescaled and similarity transformed into a trace preserving map, where $\tau_{x}=\sum_{\beta_i=1}^\mu\Bar{B}_{x,\beta_i}\otimes B_{x,\beta_i}, x\in X_i, i\in\{1,\cdots,N\}$.
\end{lemma}

\begin{proof}
To show that the c-LPS transfer operator is trace preserving is equivalent to show that the identitiy $\Vec{I}$ is the fixed point of $\tau^\dagger$, i.e., $\tau^\dagger\Vec{I}=\Vec{I}$. Assume that $\tau$ is not trace preserving. Without loss of generality assume that the two eigenvalues of $\tau$ with greatest magnitude, $\lambda_1, \lambda_2$, satisfy $|\lambda_1|>|\lambda_2|$. If we replace $B_{x_i,\beta_i},x_i\in \{1,\cdots,d\}$ with $B_{x_i,\beta_i}/\sqrt{\lambda_1}$, we obtain a transfer operator $\tau'$ that enjoys the leading eigenvalue of magnitude 1. Note that this rescaling leaves the joint probability distributions unchanged. In this case,
the quantum Perron-Frobenius theorem \cite[chapter 16]{baez2018quantum} implies that $\lambda_1=1$ and  $\tau'^\dagger$ has a unique fixed-point operator $\Vec{\sigma}_*$ (i.e., $\tau'^\dagger\Vec{\sigma}_*=\Vec{\sigma}_*$) which is the vectorization of a full-rank (and so, invertible)
positive density matrix. 

Now, we replace $B_{x_i,\beta_i},x_i\in \{1,\cdots,d\}$ with $B'_{x_i,\beta_i}=\sigma_*^{1/2}B_{x_i,\beta_i}\sigma_*^{-1/2}$. We show that $\tau''=\sum_{x_i=1}^N\tau'_{x_i}$ is a trace preserving map. Note that since ${\sigma}_*, \sigma_*^{-1/2},$ and $\sigma_*^{1/2}$ are Hermitian, we have $\overline{{\sigma}_*}=({\sigma}_*)^T, \overline{{\sigma}_*^{-1/2}}=({\sigma}_*^{-1/2})^T$, and $\overline{{\sigma}_*^{1/2}}=({\sigma}_*^{1/2})^T$, respectively. For three matrices $X,Y,Z$ that $W=XYZ$, we have $(Z^T\otimes X)\Vec{Y}=\Vec{W}$. Using these properties, we have:

\begin{equation*}
    \begin{split}
       \tau^{''^\dagger}\Vec{I}&=\Bigg(\sum_{x_i=1}^d\tau_{x_i}^{'^\dagger}\Bigg)\Vec{I}= \Bigg(\sum_{x_i=1}^d \Bigg(\sum_{\beta_{i}=1}^\mu\Bar{B}_{x_i,\beta_i}^{'^\dagger}\otimes {B}_{x_i,\beta_i}^{'^\dagger}\Bigg)\Bigg)\Vec{I}\\
       &= \Bigg(\sum_{x_i=1}^d \Bigg(\sum_{\beta_i=1}^\mu(\sigma_*^{-1/2})^T\Bar{B}_{x_i,\beta_i}^\dagger(\sigma_*^{1/2})^T\otimes \sigma_*^{-1/2}B_{x_i,\beta_i}^\dagger\sigma_*^{1/2}\Bigg)\Bigg)\Vec{I}\\
       &= \Bigg((\sigma_*^{-1/2})^T\otimes\sigma_*^{-1/2}\Bigg) \Bigg(\sum_{x_i=1}^d\sum_{\beta_i=1}^\mu\Bar{B}_{x_i,\beta_i}^\dagger\otimes B_{x_i,\beta_i}^\dagger\Bigg)\Bigg((\sigma_*^{1/2})^T\otimes\sigma_*^{1/2}\Bigg)\Vec{I}=\Bigg((\sigma_*^{-1/2})^T\otimes\sigma_*^{-1/2}\Bigg)\tau_{x_i}^{'^\dagger}\Vec{\sigma}_*\\
       &= \Bigg((\sigma_*^{-1/2})^T\otimes\sigma_*^{-1/2}\Bigg)\Vec{\sigma}_*=\Vec{I}
    \end{split}
\end{equation*}

So, the transfer operator $\tau$ can be rescaled and similarity transformed into one that is trace-preserving.
\end{proof}

\begin{proof}[Proof of Theorem \ref{thm:chmmclps}]
In order to show the equivalence, we need to show that c-LPS can be mapped to a c-HQMM, and vice versa. 

First, we show that a c-LPS can be mapped to a c-HQMM. As shown in the proof of Lemma \ref{lem:tensor2kraus}, we know that both models have operators of the same form. Lemma \ref{lem:tracepreserving} shows that the transfer operators in a c-LPS can be rescaled and similarity transformed into one that is trace-preserving. So, a c-LPS of order-$N$  with $d$-dimensional indices,  puri-rank $r$, and purification dimension $\mu$ can be rescaled and similarity transformed into a c-HQMM of $N$-horizon with $d$-dimensional and the Kraus-rank $\mu$.

Second, we show that a c-HQMM can be mapped to a c-LPS. Lemma \ref{lem:Kraus2tensor} shows that the transfer operators in a c-HQMM can be mapped into transfer operators in a c-LPS. So, a c-HQMM of $N$-horizon with $d$-dimensional and the Kraus-rank $\beta$ can be mapped into a c-LPS of order-$N$  with $d$-dimensional indices,  puri-rank $d$, and purification dimension $\beta$. 
\end{proof}

The proof structure can be directly specialized to HQMM, where we can obtain the following result:
\begin{lemma}\label{lem:HQMMvsLPS}
HQMM model is equivalent to a LPS structure where the decomposition entries $A_{i,x}^{b,a_1,a_2}$ are complex, and the $r\times r$ matrices formed by  $A_{i,x}^{b,\cdot,\cdot}$ for all $i,x,b$ are positive semi-definite (p.s.d.) for $i\in \{1, \cdots, N\}$. Further, $r\times r$ matrix $A_0$ is p.s.d.  and $r\times r$ matrix $A_{N+1}$ is the identity matrix.   
\end{lemma}

\begin{proof}
The given LPS, as defined in section \ref{sec:related}, has entries of the form:
\begin{equation*}
	\begin{split}
		T_{x_1,\cdots,x_N} &= \sum_{\{\alpha_i,\alpha'_i\}_{i=0}^{N}=1}^r\sum_{\{\beta_i\}_{i=1}^N=1}^{\mu}A_0^{\alpha_0,\alpha'_0}A_{1,x_1}^{\beta_1,\alpha_0,\alpha_1}\overline{A_{1,x_1}^{\beta_1,\alpha'_0,\alpha'_1}}A_{2,x_2}^{\beta_2,\alpha_1,\alpha_2}\overline{A_{2,x_2}^{\beta_2,\alpha'_1,\alpha'_2}}\\ &\cdots A_{N-1,x_{N-1}}^{\beta_{N-1},\alpha_{N-2},\alpha_{N-1}}\overline{A_{N-1,x_{N-1}}^{\beta_{N-1},\alpha_{N-2},\alpha_{N-1}}} A_{N,x_{N}}^{\beta_N,\alpha_{N-1},\alpha_N}\overline{A_{N,x_{N}}^{\beta_N,\alpha'_{N-1},\alpha'_N}}A_{N+1}^{\alpha_N,\alpha'_N}
	\end{split}
\end{equation*}
Since both models have operators of the same
form, as we showed in the proof of Lemma \ref{lem:tensor2kraus}, we can write and manipulate the joint
probability of a sequence of $N$ observations as the unnormalized probability mass function
over $N$ discrete random variables $\{X_i\}_{i=1}^N$ as follows:
\begin{equation}
\begin{split}
    p(x_1,\cdots,x_n)&=\Vec{I}^T(\sum_{\beta_N=1}^\mu\Bar{B}_{x_N,\beta_N}\otimes B_{x_N,\beta_N})\cdots(\sum_{\beta_1=1}^\mu\Bar{B}_{x_1,\beta_1}\otimes B_{x_1,\beta_1})\Vec{\rho_0}
\end{split}
\end{equation}
where $\Vec{I}$ is the right boundary (also called evaluation functional) is the vectorized version of the identity matrix, i.e., $A_{N+1}=I$, and $\Vec{\rho_0}$ is the vectorized version of the initial state, i.e., $A_0=\rho_0$. So, 
LPS models and HQMMs differ only in two things: (1) While LPS can have an arbitrary Kraus-rank evaluation functional, HQMMs are restricted to the identity evaluation functional $\Vec{I}^T$ of full Kraus rank. (2) HQMMs operators are trace-preserving.

Following the same approach used in the proof of Lemma \ref{lem:tensor2kraus}, the transfer operator of an LPS can be rescaled and similarity transformed into one that
is trace-preserving. For the uniform LPS models, the evaluation functional of this transformed model will then converge to $\Vec{I}^T$ \citep{adhikary2021quantum}. However, in a finite-horizon LPS this may not be the case. 

Therefore, HQMM model is equivalent to a LPS structure where the decomposition entries $A_{i,x}^{b,a_1,a_2}$ are complex, and the $r\times r$ matrices formed by  $A_{i,x}^{b,\cdot,\cdot}$ for all $i,x,b$ are positive semi-definite, where the evaluation functional is restricted to the vectorized of the identity matrix, i.e., $\Vec{I}^T$.
\end{proof}

\begin{remark}\label{rem1}
Non-terminating uniform LPS (uLPS) are equivalent to HQMM \citep{adhikary2021quantum}. In uLPS boundary vectors originates from density matrices of arbitrary rank. As shown in \citep{adhikary2021quantum}, the evaluation functional of uLPS can be rescaled and transformed into a one that will converge to $\Vec{I}^T$ when $N\to \infty.$ To have the equivalency of finite-horizon HQMM and LPS, we need to restrict LPS models to the evaluation functional $\Vec{I}^T$, as stated in Lemma \ref{lem:HQMMvsLPS}.
\end{remark}

Further, we note that the prior works do not relate c-HMM to tensor networks, to the best of our knowledge. In the following result, we relate the c-HMM to c-MPS. 

\begin{lemma}\label{lem:cHMMcMPS}
c-HMM model is equivalent to a c-MPS model, where the entries of each decomposition are real and non-negative. 
\end{lemma}
\begin{proof}
To map a c-HMM to a c-MPS with non-negative tensor elements, we can use the following procedure:
\begin{itemize}
    \item[\textbf{Input}:] A c-HMM, as shown in Figure \eqref{fig:inputCHMM}. 
    \item[Step 1:] Make a factor graph using transition (i.e., $B_i$ an $r\times r$ matrix corresponding to $p(X_i|X_{i-1})$, where $X_i,X_{i-1}\in\{1,\cdots,r\}$ for $i=1,\cdots,N$) and emission (i.e., $C_i$ an $r\times r$ diagonal matrix corresponding to $p(o_i|X_{i})$, where $X_i\in\{1,\cdots,r\}$ for $i=1,\cdots,N$) matrices of the given c-HMM, as shown in Figure \eqref{fig:buildfactorg}.
    \item[Step 2:] Build tensors $A_i=C_iB_i$ (i.e., $r\times r$ matrices) using the factor nodes obtained from Step 1, as shown in Figure \eqref{fig:buildtensors}.
    \item[\textbf{Output}:] A tensor ring with non-negative tensor elements corresponding to the given c-HMM, as shown in Figure \eqref{fig:outputCMPS}.
\end{itemize}
\begin{figure}[!ht] 
  \begin{subfigure}[t]{0.5\linewidth}
    \centering
    \includegraphics[width=\linewidth]{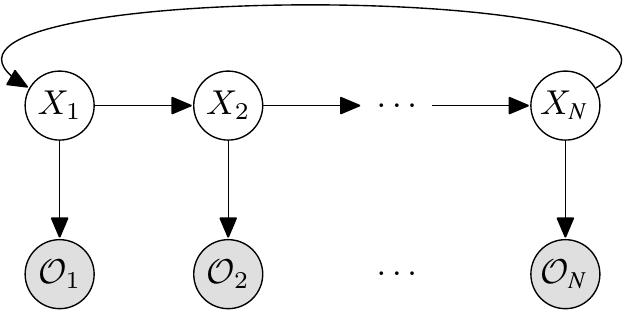} 
    \caption{\textbf{Input}: A c-HMM.} 
    \label{fig:inputCHMM} 
    \vspace{4ex}
  \end{subfigure}
  \begin{subfigure}[t]{0.5\linewidth}
    \centering
    \includegraphics[width=\linewidth]{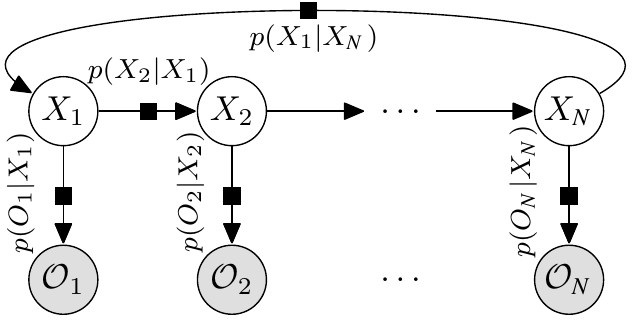} 
    \caption{Make a factor graph using transition and emission matrices.} 
    \label{fig:buildfactorg} 
    \vspace{4ex}
  \end{subfigure} 
  \begin{subfigure}[t]{0.5\linewidth}
    \centering
    \includegraphics[width=\linewidth]{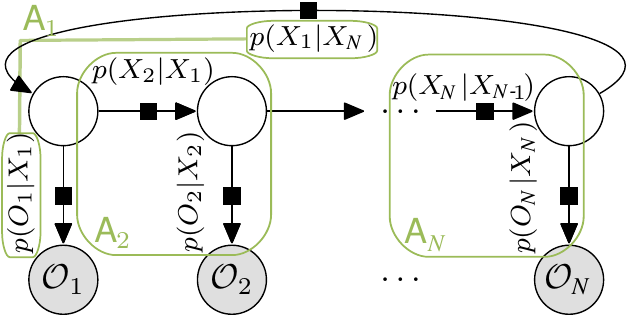} 
    \caption{Create tensors:\\ \centering$A_i=p(X_i|X_{i-1})p(\mathcal{O}_i|X_i)$.} 
    \label{fig:buildtensors} 
  \end{subfigure}
  \begin{subfigure}[t]{0.5\linewidth}
    \centering
    \includegraphics[width=\linewidth]{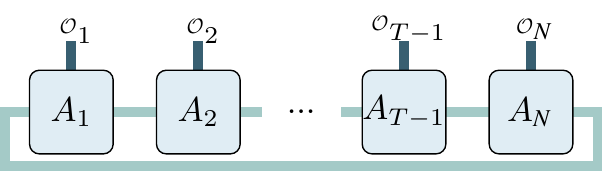} 
    \caption{\textbf{Output}: The corresponding tensor ring for the given CHMM.} 
    \label{fig:outputCMPS} 
  \end{subfigure} 
  \caption{Illustration of mapping of a c-HMM to a c-MPS with non-negative tensor elements.}
  \label{fig:mapCHMM2CMPS} 
\end{figure}

To Map a c-MPS with non-negative tensor elements to a c-HMM , we can use the following procedure:
\begin{itemize}
    \item[\textbf{Input}:] A tensor ring (c-MPS) with non-negative tensor elements, as shown in Figure \eqref{fig:inputCMPS}. 
    \item[Step 1:] Use the canonical decomposition of order-3 tensors\footnote{Canonical Polyadic Decomposition (CPD) of a third-order tensor is a minimal decomposition into a sum of rank-1 tensors.} \citep{domanov2014canonical}, to factorize each tensor in the given non-negative CMPS as follows: $A_{il}^{jk}=\sum_{s=1}^{r'} B_i^{js}C_i^{ls}D_i^{ks}$, where $r'\le\min(dr,r^2)$, as shown in Figure \eqref{fig:buildunnormfactorg}.
    \item[Step 2:] Set emission matrix elements as $p(\mathcal{O}_i=l|X_i=s)=C_i^{ls}$ and transition matrix elements as $p(X_i=s|X_{i-1}=j)=\sum_u D_{i-1}^{ju}B_i^{us}$.
    \item[Step 3:] Normalize the probabilities on every edge, as shown in Figure \eqref{fig:normfactorgCHMM}.
    \item[Step 4:] Build the normalized factor graph corresponding to the output c-HMM, as shown in Figure \eqref{fig:factorCHMM}.
    \item[\textbf{Output}:] A c-HMM corresponding to the given c-MPS, as shown in Figure \eqref{fig:outputCHMM}.
\end{itemize}
\begin{figure*}[!ht] 
  \begin{subfigure}[t]{0.5\linewidth}
    \centering
    \includegraphics[width=\linewidth]{CMPS.pdf} 
    \caption{\textbf{Input}: A c-MPS with non-negative elements.} 
    \label{fig:inputCMPS} 
    \vspace{4ex}
  \end{subfigure}
  \begin{subfigure}[t]{0.5\linewidth}
    \centering
    \includegraphics[width=\linewidth,page=1]{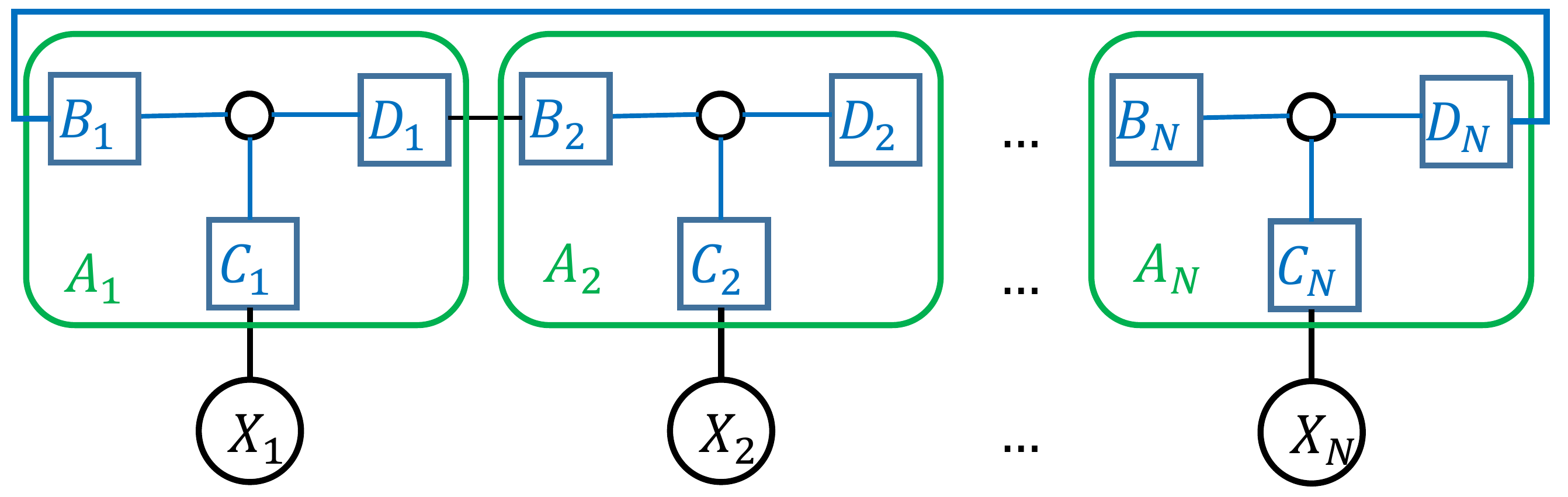} 
    \caption{Use CPD decomposition to build an unnormalized factor graph.} 
    \label{fig:buildunnormfactorg} 
    \vspace{4ex}
  \end{subfigure} 
  \begin{subfigure}[t]{0.5\linewidth}
    \centering
    \includegraphics[width=\linewidth,page=2]{mapT2H.pdf} 
    \caption{Normalize probabilities on every edge.} 
    \label{fig:normfactorgCHMM} 
  \end{subfigure}
  \begin{subfigure}[t]{0.5\linewidth}
    \centering
    \includegraphics[width=.75\linewidth,page=3]{mapT2H.pdf} 
    \caption{Normalized factor graph corresponding to the output c-HMM.} 
    \label{fig:factorCHMM} 
  \end{subfigure} 
  \centering
  \begin{subfigure}[t]{0.5\linewidth}
    \centering
    \includegraphics[width=.5\linewidth,page=4]{mapT2H.pdf}
    \subcaption{\textbf{Output}: The corresponding CHMM for the given c-MPS.} 
    \label{fig:outputCHMM} 
  \end{subfigure}
  \caption{Illustration of mapping of a c-MPS with non-negative tensor elements to a c-HMM.}
  \label{fig:mapCMPS2CHMM} 
\end{figure*} 
\end{proof}

\begin{figure*}[!ht]
    \begin{subfigure}[b]{0.32\textwidth}
    \centering
        \includegraphics[width=\textwidth]{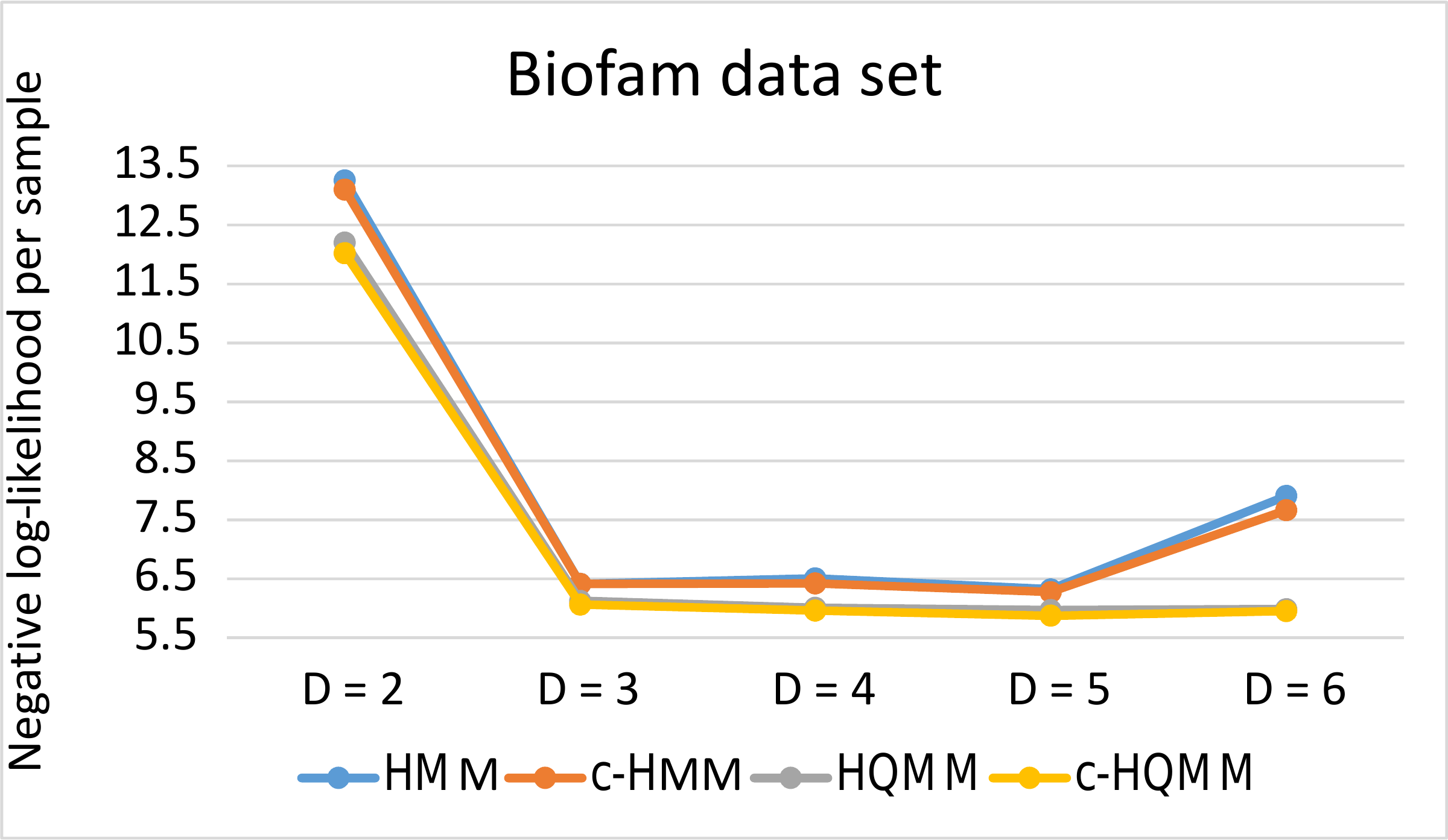}
        \caption{}
        \label{fig:biofam}
    \end{subfigure}
    \hfill
    \begin{subfigure}[b]{0.32\textwidth}
    \centering
        \includegraphics[width=\textwidth]{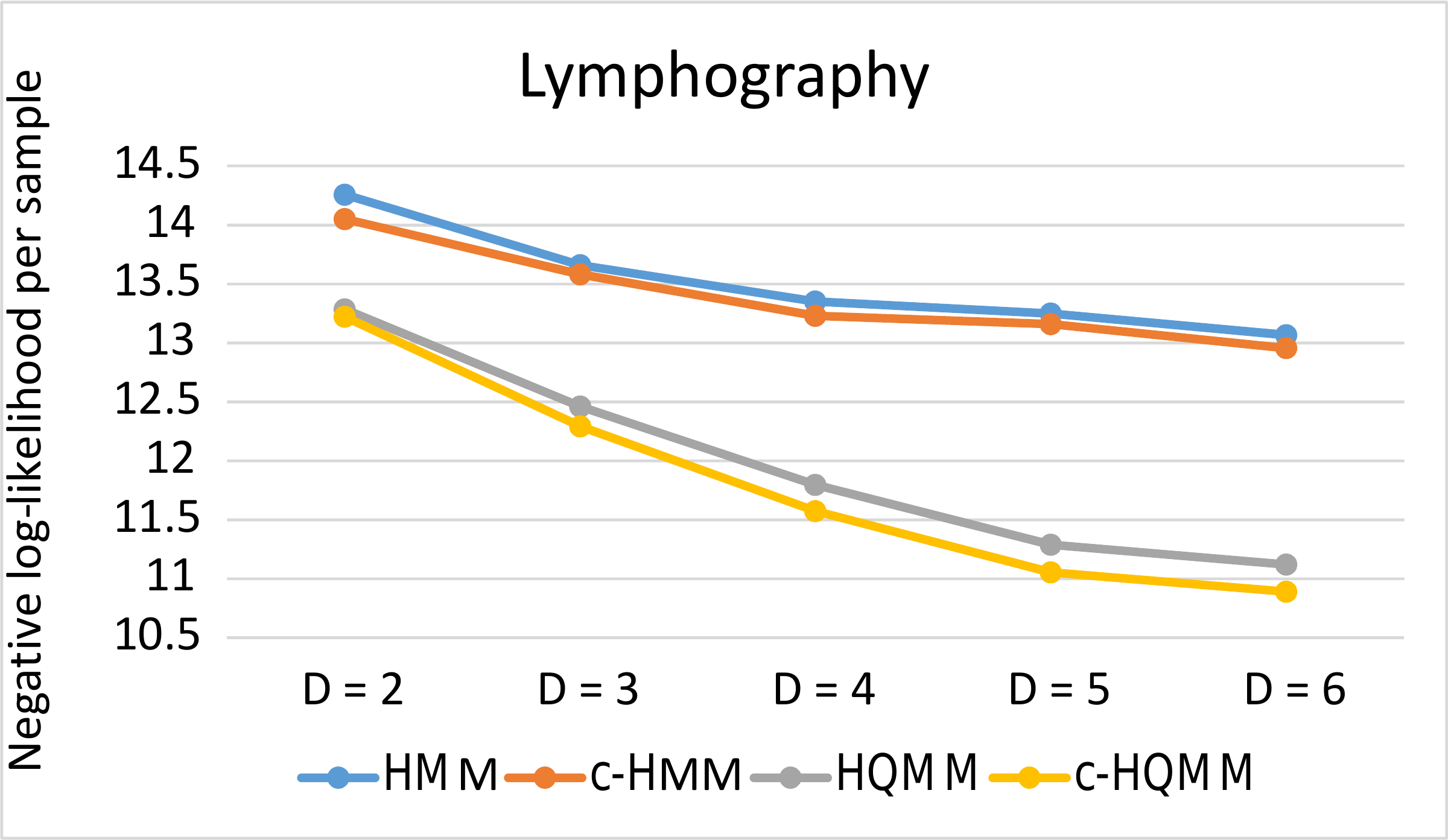}
        \caption{}
        \label{fig:lympho}
    \end{subfigure}
    \begin{subfigure}[b]{0.32\textwidth}
    \centering
        \includegraphics[width=\textwidth]{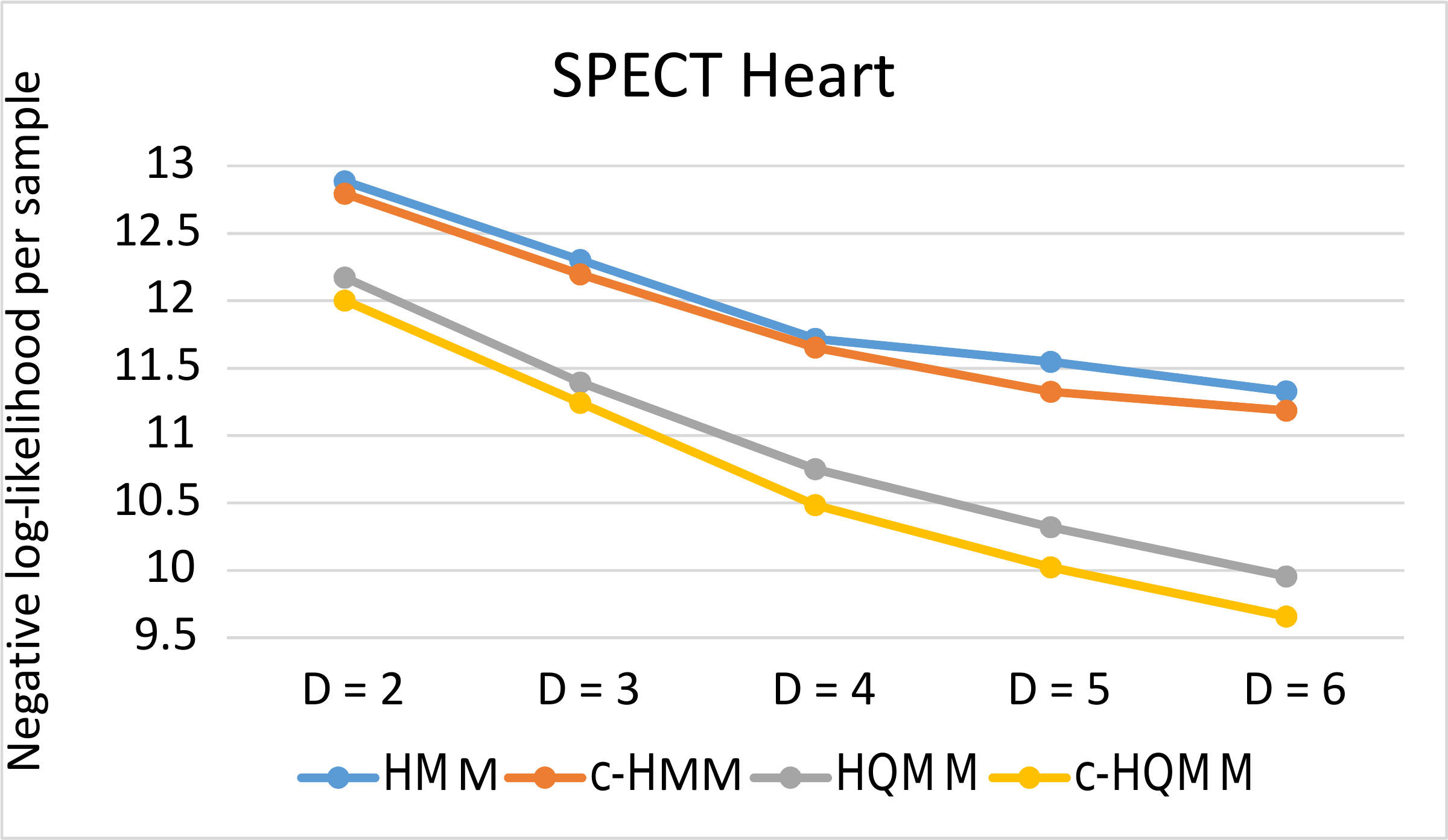}
        \caption{}
        \label{fig:spect}
    \end{subfigure}
    
    \begin{subfigure}[b]{0.32\textwidth}
    \centering
        \includegraphics[width=\textwidth]{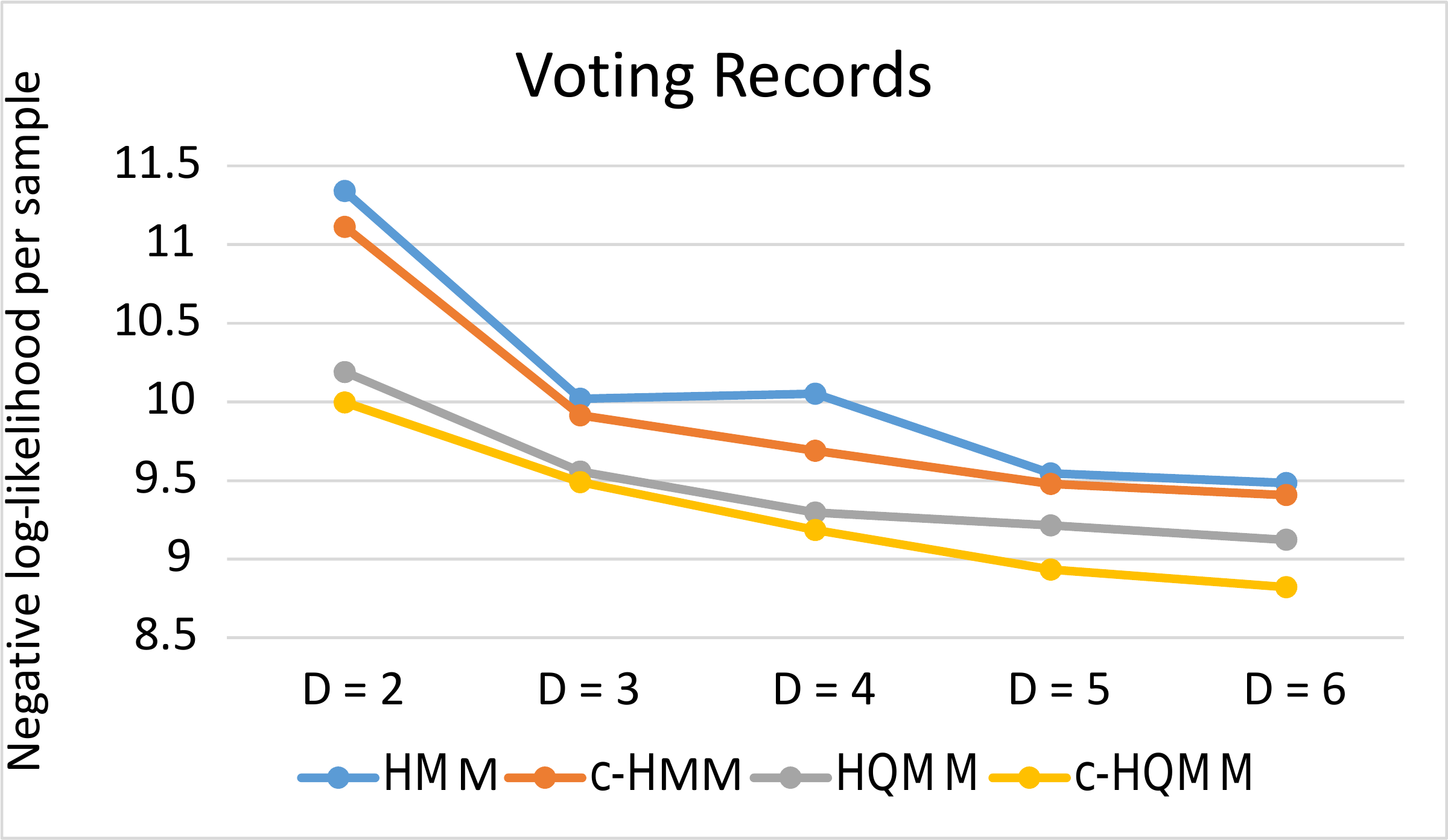}
        \caption{}
        \label{fig:voting}
    \end{subfigure}
    \hfill
    \begin{subfigure}[b]{0.32\textwidth}
    \centering
        \includegraphics[width=\textwidth]{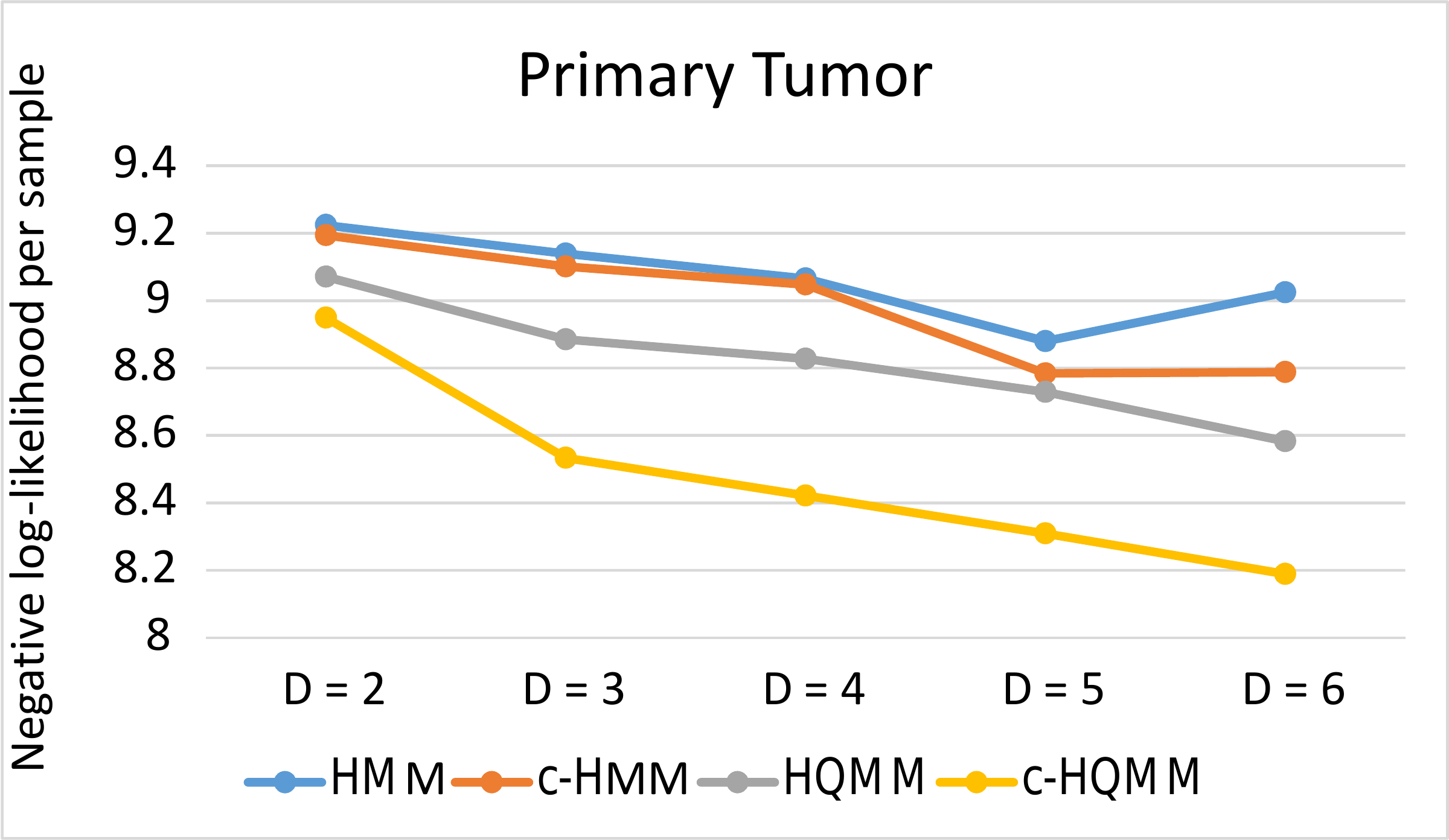}
        \caption{}
        \label{fig:tumor}
    \end{subfigure}
    \hfill
    \begin{subfigure}[b]{0.32\textwidth}
    \centering
        \includegraphics[width=\textwidth]{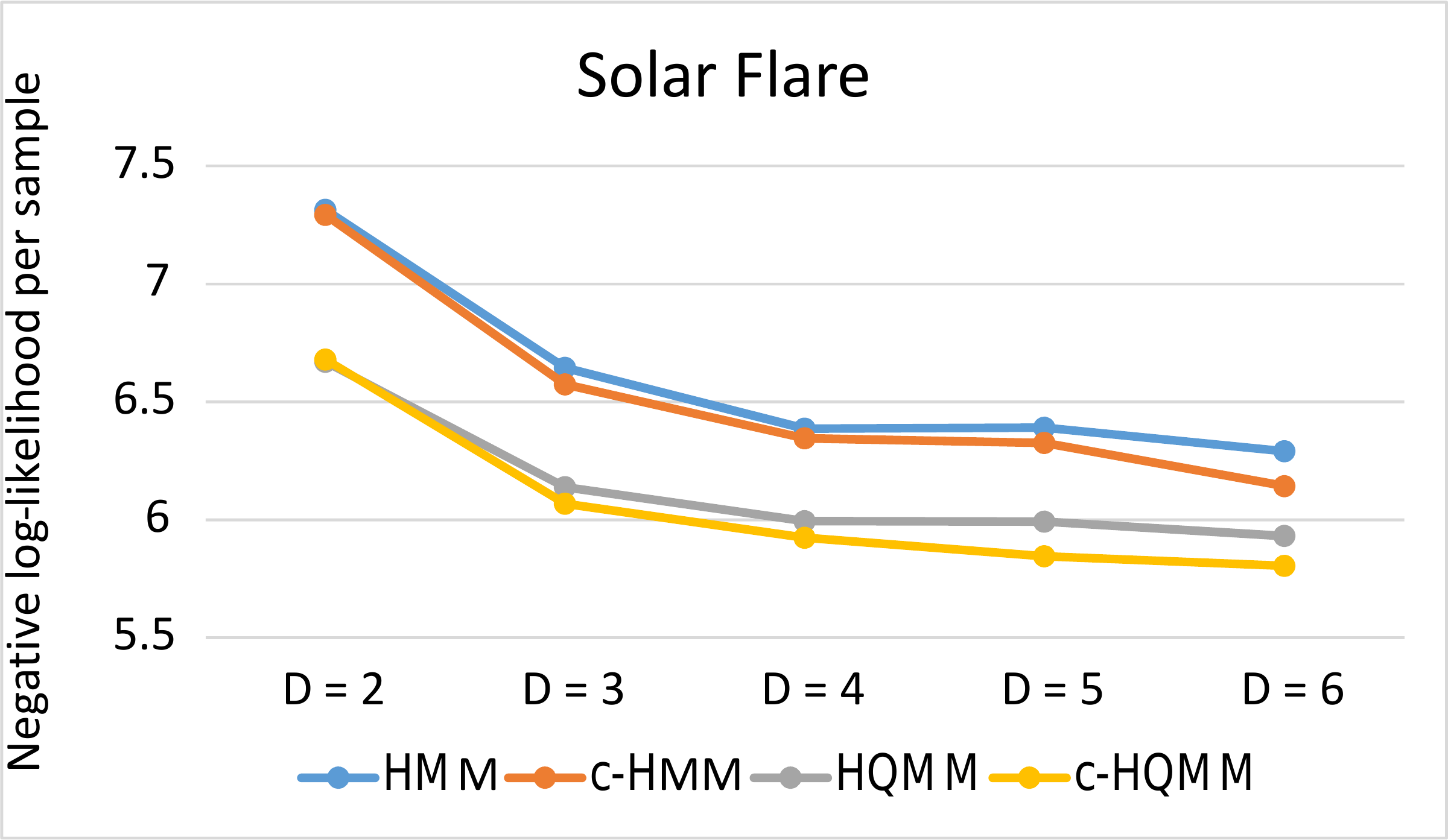}
        \caption{}
        \label{fig:flare}
    \end{subfigure}
    \caption{Maximum likelihood estimation with tensor rings MPS, c-MPS, LPS, and c-LPS for learning HMM, c-HMM, HQMM, and c-HQMM from the data on different data sets: a) biofam data set of family life states from the Swiss Household Panel biographical survey \citep{muller2007classification}; data sets from the UCI Machine Learning Repository \citep{Dua2019}: b) Lymphography, c) SPECT Heart, d) Congressional Voting Records, e) Primary Tumor, and f) Solar Flare.}
    \label{fig:evaluation}
\end{figure*}

\begin{figure}[!ht]
	\centering
	\includegraphics{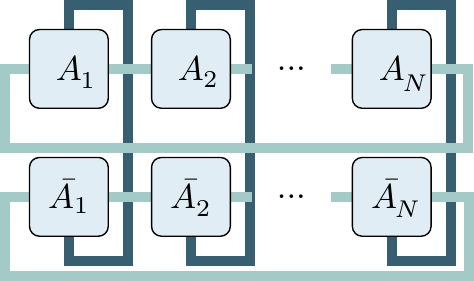}
	\caption{Contraction of tensor ring to compute:  $Z_T=\sum_{X_1,\cdots,X_N}T_{X_1,\cdots,X_N}.$}
	\label{fig:ZTcontract}
\end{figure}
\section{Learning Algorithm for Circular LPS Models}\label{sec:learning}

In this section we propose an algorithm for learning c-LPS Models as in Theorem \ref{thm:chmmclps} via a maximum likelihood estimation (MLE) approach. The proposed algorithm is a modification of the algorithm proposed in \citep{Glasser2019} for learning LPS models, except that we take into account the positive semi-definite nature of the decomposition and the cyclic structure. 

\begin{problem}[MLE for Distribution Approximation]\label{problem:mle}
	Assume that $\{\textbf{\textrm{x}}_i=(x_1^i,\cdots,x_N^i)\}_{i=1}^n$ is a sample of size $n$ from an experiment with $N$ discrete random variables. To estimate this discrete multivariate distribution, we use c-LPS model as defined in section \ref{sec:clps}. So, we have:
	\begin{equation*}
		\begin{split}
			p({x_1,\cdots,x_N}) & \approxeq\sum_{\{\alpha_i,\alpha'_i\}_{i=1}^N=1}^r\sum_{\{\beta_i\}_{i=1}^N=1}^{\mu} A_{1,x_1}^{\beta_1,\alpha_{N},\alpha_1}\overline{A_{1,x_1}^{\beta_1,\alpha'_{N},\alpha'_1}}\cdots A_{N,x_{N}}^{\beta_N,\alpha_{N-1},\alpha_{N}}\overline{A_{N,x_{N}}^{\beta_N,\alpha'_{N-1},\alpha'_{N}}},
		\end{split}
	\end{equation*}
	where the tensor decomposition entries follow the structure in Theorem \ref{thm:chmmclps}. Our objective here is to estimate tensor elements of the c-LPS, i.e., $w=A_{i,x_i}^{\beta_i,\alpha_{L_i},\alpha_{R_i}}$, for $i=1,\cdots,N$. For this purpose, we minimize the negative log-likelihood:
	\begin{equation}\label{eq:loglik}
		L=-\sum_i \log\frac{T_{\textbf{\textrm{x}}_i}}{Z_T}
	\end{equation}
	where $T_{\textbf{\textrm{x}}_i}$ is given by the contraction of c-LPS, and $Z_T=\sum_{{\textbf{\textrm{x}}_i}}T_{\textbf{\textrm{x}}_i}$ is a normalization factor.
\end{problem}


To find the optimal solution, we calculate the derivative of the log-likelihood with respect to $w$ as follows:
\begin{equation}\label{eq:derivative}
	\partial_wL=-\sum_i \frac{\partial_wT_{\textbf{\textrm{x}}_i}}{T_{\textbf{\textrm{x}}_i}}-\frac{\partial_wZ_T}{Z_T}
\end{equation}
We use a mini-batch gradient-descent algorithm to minimize the negative log-likelihood. At
each step of the optimization, the sum is computed over a batch of training instances. The parameters
in the tensor network are then updated by a small step in the inverse direction of the gradient. To satisfy the condition in Theorem \ref{thm:chmmclps}, we project the $r\times r$ matrices formed by  $A_{i,x}^{b,\cdot,\cdot}$ for all $i,x,b$ to the positive semi-definite (p.s.d.) matrices using the standard Singular Value Decomposition (SVD) method \citep{dattorro2010convex}.  Note that we use Wirtinger derivatives with respect to the conjugated tensor elements. Now, we explain how to compute $\partial_wZ_T$, $Z_T$,  $\partial_wT_{\textbf{\textrm{x}}_i}$, and $T_{\textbf{\textrm{x}}_i}$ in equation \eqref{eq:derivative}. For a c-LPS of puri-rank $r$, the normalization $Z_T$ can be computed by contracting the tensor network:
\begin{equation}\label{eq:contractZT}
	Z_T=\sum_{x_1,\cdots,x_N}T_{x_1,\cdots,x_N} 
\end{equation}

This contraction is performed, as shown in Figure \eqref{fig:ZTcontract}, from left to right by contracting at each
step the two vertical indices (corresponding to $d_i$ and $\bar{d}_i$ with respect to the supports of $X_i$ and $\bar{X}_i$) and then each of the two horizontal indices (with respect to $\alpha_i$s and $\alpha'_i$s, respectively). Finally, we trace out the indices corresponding to the rings. In this contraction,
intermediate results from the contraction of the first $i$ tensors are stored in $E_i$, and the same procedure
is repeated from the right with intermediate results of the contraction of the last $N-i$ tensors stored in $F_{i+1}$. The derivatives of the normalization for each tensor are then computed as
\begin{center}
	\begin{tcolorbox}[hbox]
		$
		\begin {aligned}
		\frac{\partial_wZ_T}{\partial_w\bar{A}_{i,m}^{j,k,l}}
		\qquad = \qquad
		\raisebox{-7mm}{\includegraphics[scale=.25]{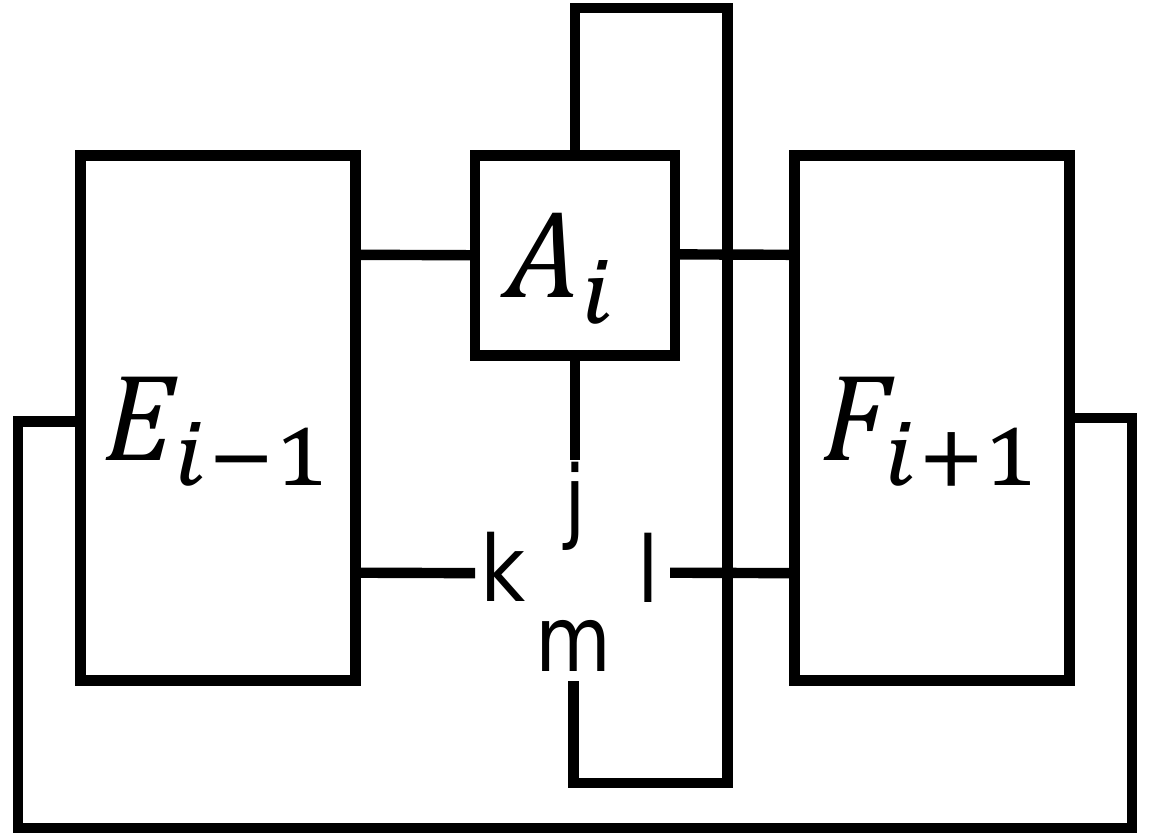}}
		\end {aligned}
		$
	\end{tcolorbox}
\end{center}

Computing $T_{\textbf{\textrm{x}}_i}$ for a training sample and its derivative is
done in the same way, except that the contracted index corresponding to an observed variable is now fixed to its observed value. We note that a similar approach to learn the model can be used for HQMM, HMM, and c-HMM structures. 

\section{Numerical Evaluations: Maximum Likelihood Estimation on Real Data}\label{sec:experiments}
To evaluate the performance of the proposed algorithm for learning c-HQMMs, we used the same datasets as used in  \citep{Glasser2019} and learn HMM, c-HMM, HQMM, and c-HQMM using their respective tensor representations. HMM is equivalent to  MPS$_{\mathbb{R}_{\ge 0}}$, and is a baseline for other structures. We note that equivalence of both c-HMMs and c-HQMMs to tensor networks first appear in this paper. Further, the equivalence of LPS and HQMM for finite $N$ with non-uniform Kraus operators is also studied for the first time in this paper. We compare the performance of training HMM, c-HMM, HQMM, and c-HQMM using equivalent tensor representations on  six different real data of categorical variables, where following parameters are used:


\begin{itemize}
\itemsep0em
    \item Bond dimension/rank of the tensor networks: $r=2, 3, 4,5 \textrm{ and } 6$. 
    \item Learning rate was chosen using a grid search on powers of 10 going from $10^{-5}$ to $10^5$.
    \item Batch size, i.e., the number of training samples per minibatch, was set to 20.
    \item Number of iterations was set to a maximum of 1000. 
    \item The dimension of the purification index, i.e., $\mu$ for LPS and c-LPS was set to 2.
\end{itemize}

Each data point reported here is the lowest negative log-likelihood obtained from 10 trials with different initialization of tensors.

\textbf{Results:} The obtained results, summarized in Figure \ref{fig:evaluation}, show that (1) The tensor representations can be used to learn different HMMs. (2) We observe
that despite the different algorithm choice, on almost all data sets, c-LPS and LPS lead to better modeling of the data distribution  for the same rank as compared to MPS. (3) The results indicate that c-HQMM outperforms HQMM, HMM, and c-HMM. (4) In many cases, the performance difference between LPS and c-LPS for rank 5 and 6 is more significant than the cases with the rank 2 and 3. Further, the improvement depends on the dataset. We also note that we plot negative of log likelihoods, so the gap in the likelihoods is larger. 
The results suggest that in generic settings HQMM and c-HQMM should be preferred over both HMM and c-HMM models, respectively. Further, c-HQMM gives the best performance among the considered models.  

\section{Conclusion}
This paper proposes a new class of hidden Markov models, that we called circular  Hidden Quantum Markov Models (c-HQMMs).
c-HQMMs can be used to model temporal data in quantum datasets (with classical datasets as a special case). We proved that c-HQMMs are  equivalent to circular LPS models with positive-semidefinite constraints on certain matrix structure in the LPS decomposition. Leveraging this result, we proposed an MLE based algorithm for learning c-HQMMs from data via c-LPS. We evaluated the proposed learning approach  on six real datasets, demonstrating the advantage of c-HQMMs on multiple datasets as compared to  HQMMs, circular HMMs, and HMMs. 

\section*{Acknowledgements}
 This research was supported by the Defense Advanced Research Projects Agency (DARPA) Quantum Causality [Grant No. HR00112010008].

\bibliographystyle{IEEEtran}
\bibliography{references}
\end{document}